\documentclass[a4paper,11pt]{article}

\usepackage[utf8]{inputenc}
\usepackage[T1]{fontenc}
\usepackage{lmodern}
\usepackage{microtype}
\usepackage{graphicx}
\usepackage{booktabs,tabularx}

\usepackage[margin=1in]{geometry}

\usepackage{mathtools}
\usepackage{amssymb}
\usepackage{amsthm}
\usepackage{complexity}
\numberwithin{equation}{section}
\usepackage{xcolor}

\usepackage{hyperref}

\hypersetup{colorlinks=true, linkcolor=blue, citecolor=blue, urlcolor=blue}
\usepackage[capitalise]{cleveref}

\newtheorem{theorem}{Theorem}[section]
\newtheorem{lemma}[theorem]{Lemma}

\newtheorem{corollary}[theorem]{Corollary}

\newtheorem{claim}[theorem]{Claim}

\theoremstyle{definition}
\newtheorem{definition}[theorem]{Definition}
\newtheorem{remark}[theorem]{Remark}
\newtheorem{example}[theorem]{Example}

\newenvironment{claimproof}{\begin{proof}[Proof of Claim]}{\end{proof}}

\crefname{claim}{Claim}{Claims}
\Crefname{claim}{Claim}{Claims}

\newcommand\pfun{\rightharpoonup}

\newcommand{\Bbound}{\ensuremath{B}}
\newcommand{\weight}{\omega}
\newcommand{\st}{\ensuremath{\mid}}
\newcommand{\cost}{\mathrm{cost}}
\newcommand{\costGED}{\mathrm{cost}}

\renewcommand{\phi}{\varphi}
\renewcommand{\epsilon}{\varepsilon}
\newcommand{\ged}{\ensuremath{\delta_{\mathsf{edit}}}}

\newcommand{\VC}{{\textup{VC}}}

\newcommand{\Bigmid}{\mathrel{\Big|}}
\newcommand{\ceil}[1]{\left\lceil#1\right\rceil}
\newcommand{\floor}[1]{\left\lfloor#1\right\rfloor}

\renewcommand{\tilde}{\widetilde}
\renewcommand{\hat}{\widehat}
\renewcommand{\bar}{\overline}

\newcommand{\symdiff}{\mathbin{\triangle}}

\usepackage[backend=bibtex,style=numeric,maxnames=99]{biblatex}
\renewbibmacro*{doi+eprint+url}{%
  \iftoggle{bbx:url}
  {\iffieldundef{doi}{\usebibmacro{url+urldate}}{}}
  {}%
  \newunit\newblock
  \iftoggle{bbx:eprint}
  {\usebibmacro{eprint}}
  {}%
  \newunit\newblock
  \iftoggle{bbx:doi}
  {\printfield{doi}}
  {}%
}

\newcommand{\defpromiseproblem}[4]{
  \vspace{1mm}
  \noindent\fbox{
  \begin{minipage}{0.96\textwidth}
  \begin{tabular*}{\textwidth}{@{\extracolsep{\fill}}lr} #1 \\ \end{tabular*}
  {\bf{Input:}} #2  \\
  {\bf{Promise:}} #3  \\
  {\bf{Question:}} #4
  \end{minipage}
  }
  \vspace{1mm}
}

\newcommand{\WL}[2]{\chi_{(\infty)}^{#1,#2}}
\newcommand{\WLit}[3]{\chi_{(#2)}^{#1,#3}}

\begin{document}

\title{Robust Graph Isomorphism,\\Quadratic Assignment and VC Dimension}
\author{Anatole Dahan, Martin Grohe, Daniel Neuen, Tomáš Novotný}
\maketitle

\begin{abstract}
We present an additive $\epsilon n^{2}$-approximation algorithm for the Graph
Edit Distance problem (GED) on graphs of~VC~dimension $d$ running in time
$n^{O(d/\epsilon^{2})}$.
In particular, this recovers a previous result by Arora, Frieze, and Kaplan
[Math.\ Program.\ 2002] who gave an~\mbox{$\epsilon n^{2}$-approximation}
running in time $n^{O(\log n/\epsilon^{2})}$.

Similar to the work of~Arora et al., we extend our results to arbitrary
Quadratic Assignment problems (QAPs) by introducing a notion of~VC~dimension
for QAP instances, and giving an $\epsilon n^{2}$-approximation for QAPs with
bounded weights running in time $n^{O(\epsilon^{-2}(d + \log\epsilon^{-1}))}$.

As a particularly interesting special case, we further study the problem
$\epsilon$-$\GI$, which entails determining if two graphs $G,H$ over $n$
vertices are isomorphic, when promised that if they are not, their graph edit
distance is at least $\epsilon n^{2}$.
We show that the standard Weisfeiler--Leman algorithm of~dimension
$O(\epsilon^{-1}d\log(\epsilon^{-1}))$ solves this problem on graphs of
VC~dimension $d$.
We also show that dimension $O(\epsilon^{-1}\log n)$ suffices on arbitrary
$n$-vertex graphs, while $k$-WL fails on instances at distance
$\Omega(n^{2}/k)$.
\end{abstract}

\section{Introduction}

Computing the~distance between two graphs is a~notoriously hard
algorithmic problem
\cite{ArvindKKV12,GervensG22,GroheRW18,KollaKMS18}. One of~the~most
natural distance measures is the~\emph{graph edit distance}, defined to~be
the~minimum number of~edges that have to~be flipped (added or deleted) in
one of~the~graphs to~make it isomorphic to~the~other graph (assuming
both graphs have the~same number of~vertices). Grohe, Rattan, and
Woeginger~\cite{GroheRW18} proved that the~problem of~computing the
edit distance between two graphs of~the~same order (denoted by GED in the
following) is NP-hard even if both input graphs are trees, and~Arvind,
Köbler, Kuhnert, and~Vasudev~\cite{ArvindKKV12} proved that there is
no PTAS for GED unless $\P =  \NP$. However, there is a~very interesting and
nontrivial upper bound. Arora, Frieze, and~Kaplan \cite{AroraFK02} give an
additive approximation algorithm for GED running in
quasipolynomial time: given $\epsilon>0$ and~two graphs $G,H$ of~order
$n$, the~algorithm approximates the~edit distance up to~an additive
error of~$\epsilon n^2$ in time $n^{O(\epsilon^{-2}\log n)}$. Note
that such an approximation is only interesting in the~\emph{dense regime}
of~graphs with $\Omega(n^2)$ edges.

We parametrise the~edit-distance problem by the~VC~dimension of~the
input graphs, where we define the~VC~dimension of~a graph $G$ to~be
the~VC~dimension of~the~neighbourhood set system
$\mathcal N_G\coloneqq\{ N_G(v)\mid v\in V(G)\}$.
As our first main result, we show that edit distance of~graphs
of~bounded VC~dimension can be approximated up to~an additive error
of~$\epsilon n^2$ in~polynomial~time.
\begin{theorem}
    \label[theorem]{thm:main-edit-distance}
    There is an algorithm that, given $\varepsilon > 0$ and~$n$-vertex
    graphs $G,H$ of~VC~dimension at~most $d$, outputs a~bijection
    $\phi\colon V(G) \to V(H)$ with $\cost(\phi) \leq \ged(G,H) + \epsilon n^2$
    in time $n^{O\left(\epsilon^{-2}d\right)}$.
\end{theorem}

Here, $\ged(G,H)$ denotes the~graph edit distance between $G$ and~$H$,
and $\cost(\phi)$ denotes the~number of~edge mismatches between $G$
and $H$ under the~bijection $\phi$. Note that the~algorithm interpolates
smoothly between this case of~bounded VC~dimension and~Arora et al.'s
quasipolynomial time algorithm for arbitrary graphs \cite{AroraFK02},
as the~VC~dimension of~an $n$-vertex graph is at~most $\log n$.

In fact, our algorithm even works for edge-weighted graphs (with
rational edge weights bounded in their absolute value).  Here, an
isomorphism also needs to~preserve edge weights, and~changing the
weight of~an edge $e$ from $w$ to~$w'$ has cost $|w - w'|$ (where
non-edges are assumed to~have weight $0$).  Also, to~generalise the
notion of~VC~dimension, we employ the~natural approach of~defining
a neighbourhood set system for each weight
threshold $t \in \mathbb Q$.  More precisely, for every
$t \in \mathbb Q$, we require that the~set system
$\mathcal N_G^{t}\coloneqq\{ N_G^{t}(v)\mid v\in V(G)\}$ has VC
dimension at~most $d$, where $N_G^{t}(v)$ denotes all neighbours of
$v$ reachable via an edge of~weight $> t$.\footnote{ One might also be
  tempted to~consider the~set system
  $\mathcal N_G^* \coloneqq \{ N_G^{t}(v)\mid v\in V(G), t \in \mathbb
  Q\}$.  However, note that this is a~more restrictive definition: If
  $\mathcal N_G^*$ has VC~dimension $d$, then in particular
  $\mathcal N_G^{t}$ has VC~dimension at~most $d$ for every
  $t \in \mathbb Q$.}

While VC~dimension has played an important role in algorithmic
geometry at~least since the~late 1980s (see, for example,
\cite{Matousek99}), graph algorithms exploiting bounded VC~dimension
have only started to~appear fairly recently
\cite{ChanCGKLZ25,DucoffeHV22,DurajKP24,LeW24}. In this context, it is
worth mentioning that the~monadically dependent graph classes,
conjectured to~be the~limit for the~fixed-parameter tractability of
first-order model checking, can be defined in terms of~bounded VC
dimension of~set systems definable in first-order logic
(see, for example, \cite[Conjecture~1]{DreierEMMPT24}). However, all
these results and~the~conjecture are based on
more general set systems than just the~neighbourhood set system
$\mathcal N_G$. We are not aware of~algorithmic results just
exploiting a~bounded VC~dimension of~$\mathcal N_G$.
However, many other graph parameters that have received a~lot of
attention in the~graph theory literature, such as maximum degree,
tree width, clique width \cite{CourcelleO00}, twin width \cite{BonnetKTW22}
or merge width \cite{DreierT25}, can be bounded in terms of~VC
dimension, so algorithms for graph classes of~bounded VC~dimension
also apply to~classes where any one of~these parameters is bounded.

Actually, Arora et al.'s approximation algorithm \cite{AroraFK02}
is not just for edit distance, but for the~much wider family of~quadratic
assignment problems (QAP). In an assignment problem, we generally wish
to~minimise a~cost function over the~set of~all permutations on
the~set $[n] \coloneqq \{1,\dots,n\}$. Each permutation can be represented by 
$n^2$ variables $x(v,v')$ such that
\begin{align}
\label{eq:qap1}
 \sum_{v'} x(v,v') &=\sum_{v} x(v,v')= 1 &&\forall v,v' \in [n], \\
 \label{eq:qap2}
 x(v,v') &\geq 0 &&\forall v,v' \in [n].
\end{align}
An (integral) solution $x$ satisfying these constraints is called
a \emph{perfect matching}. In an instance of~QAP, we are given rational
coefficients $c(v,v',w,w')\in \mathbb Q$ for all $v,v',w,w'$ and~we wish
to~find a~perfect matching $x$ that minimises
\[\cost(x) = \sum_{v,v',w,w'} c(v,v',w,w') \cdot x(v,v') \cdot x(w,w').\]
It is well-known that GED is a~special case of~QAP \cite{AroraFK02,ArvindKKV12};
see~\cref{ex:ged_qap}.

To generalise our results to~QAP, we introduce a~notion of~VC~dimension
for QAP instances. Basically, the~idea is to~view the~coefficients
$c(v,v',w,w')$ as edge weights in a~graph with vertex set $[n] \times [n]$,
and~define the~VC~dimension of~the~QAP instance to~be the~VC~dimension of~the
resulting edge-weighted graph. We stress that for graph problems such as GED,
the~VC~dimension of~the~corresponding QAP instance coincides with the
VC~dimension of~the~underlying (edge-weighted) graphs (up to~a constant factor).
We obtain the~following result that extends the~results
of~Arora et al.\ \cite{AroraFK02}. A QAP instance is \emph{$\Bbound$-bounded}
if $|c(v,v',w,w')| \leq \Bbound$ for all $v,v',w,w' \in [n]$.

\begin{theorem}
    \label{thm:quadratic-assignment-bounded-vc-approximation-main}
    There is an algorithm that, given $\varepsilon > 0$ and~a $\Bbound$-bounded
    QAP instance of~VC~dimension at~most $d$, outputs a~perfect matching $x$
    such that $\cost(x) \leq \cost(x^*) + \epsilon n^2$ in~time
    $n^{O\left( \Bbound^2\epsilon^{-2}(d + \log \Bbound\epsilon^{-1})\right)}$,
    where $x^*$ denotes an optimum solution.
\end{theorem}

Going back to~graphs, a~particularly interesting special case of~GED
is the~\emph{robust graph isomorphism} problem. For every $\epsilon>0$,
consider the~following promise problem $\epsilon$-$\GI$: given two
$n$-vertex graphs $G,H$ such that either $G$ and~$H$ are isomorphic
or $G$ and~$H$ have edit distance at~least $\epsilon n^2$, decide
if $G$ and~$H$ are isomorphic. Note that this is a~version of~robust
graph isomorphism for the~dense regime. For the~sparse regime where
the~promise is that the~edit distance is either $0$ or at~least
$\epsilon(n+m)$, where $m$ denotes the~number of~edges of~the~input
graph, the~problem has been studied by O'Donnell, Wright, Wu, and
Zhou~\cite{OdonnellWWZ14}, and~they obtained various hardness results.

It is an open problem whether $\epsilon$-$\GI$ can be solved in
polynomial time for every fixed $\epsilon>0$; we remark that O'Donnell
et al.'s lower bounds \cite{OdonnellWWZ14} do not yield direct insights
for the~dense regime. It follows from \cite{AroraFK02} that $\epsilon$-$\GI$
can be solved in time $n^{O(\epsilon^{-2}\log n)}$, slightly improving
over Babai's quasi-polynomial isomorphism test \cite{Babai16}
(which solves the~problem without making use of~the~promise). More generally,
\cref{thm:main-edit-distance} implies that $\epsilon$-$\GI$ can be
solved in time $n^{O(\epsilon^{-2}d)}$ on~graphs of~VC~dimension at~most $d$.
Observe that this regime is interesting since the~restriction of~(exact)
graph isomorphism is known to~be hard (as hard as the~general graph
isomorphism problem) for graph classes of~bounded expansion and~bounded
twin width, both of~which implies bounded VC~dimension.

We show that both runtime bounds can be slightly improved, and in fact this can be achieved by a~much simpler
algorithm than for \cref{thm:main-edit-distance}. The
\emph{Weisfeiler-Leman (WL) algorithm}, going back to~\cite{WeisfeilerL68},
is a~simple combinatorial algorithm originally designed as an (incomplete)
isomorphism test, and~which is commonly used for isomorphism testing
(see, e.g., \cite{Babai16,cai_optimal_1992,GroheN21,Kiefer20}).
Additionally, it has recently received a~lot of~attention due to~a number
of~surprising connections with combinatorial
optimisation \cite{AtseriasM13,GroheO15,Malkin14}, graph homomorphisms
\cite{Dvorak10,DellGR18}, and~even machine learning
\cite{MorrisRFHLRG19,ShervashidzeSLMB11,XuHLJ19}.
We prove that the~$O(\epsilon^{-1}d\log(\epsilon^{-1}))$-dimensional
WL~algorithms solves $\epsilon$-$\GI$ on~graphs of~VC~dimension $d$.
Also, for general $n$-vertex graphs, we show that
$O(\epsilon^{-1}\log n)$-dimensional WL~algorithm solves $\epsilon$-$\GI$.
Since the $k$-dimensional WL algorithm can be implemented in time $O(k^2 n^{k+1}\log n)$ \cite{ImmermanL90}, this implies that
$\epsilon$-$\GI$ can be solved in time $n^{O(\epsilon^{-1}d\log(\epsilon^{-1}))}$ on~graphs of~VC~dimension at~most $d$,
and in time $n^{O(\epsilon^{-1}\log n)}$ on general graphs.

Our upper bound on the WL dimension are in stark contrast to~well-known results stating that the
$\Omega(n)$-dimensional WL~algorithm is necessary to~solve graph
isomorphism \cite{cai_optimal_1992}. In fact, the~$\Omega(n)$ lower bound
for the~WL~dimension even holds for graphs of~maximum degree $3$
\cite{cai_optimal_1992}, which in particular have constant VC~dimension.
For coloured graphs, we also give an upper bound of~$O(\epsilon^{-1}s)$
on the~WL~dimension of~$\epsilon$-$\GI$, where $s$ denotes the~maximum
colour class size. Again, this is contrasted by the~fact that $\Omega(n)$-WL
is necessary to~solve graph isomorphism on~graphs
of~colour class size $4$ \cite{cai_optimal_1992}.

We complement these results by a~lower bound stating that
stating that for every $\epsilon$ there is a~$k=\Omega(\epsilon^{-1})$ such that
$k$-dimensional WL does not solve $\epsilon$-$\GI$. Moreover, for graphs
of~bounded VC~dimension and~for graphs of~bounded colour class size,
we obtain tight polynomial bounds: the~smallest $k$ such
that $k$-WL solves $\epsilon$-$\GI$ on~these respective classes
is polynomial in $\epsilon$.

\section{Preliminaries}
\label{sec:preliminaries}

\newcommand{\graph}{\mathrm{graph}}
\newcommand{\funct}{\mathrm{func}}

For $a,b\in\mathbb R$, we denote by~$[a, b]$  the~set
$\{ x\in\mathbb R\st a \leq x \leq b\}$ and~by~$[a\pm b]$ the~set $[a-b,a+b]$.
For sets $X,Y$, we denote their symmetric difference by $X\symdiff Y \coloneqq (X \cup Y) \setminus (X \cap Y)$.
We often consider partial injective functions $\alpha\colon X\to Y$ between two sets $X,Y$.
Depending on the~context, it may be convenient
to view them as sets of~pairs $\alpha\subseteq X\times Y$. We use
the~notations $\alpha(x)=y$ and~$(x,y)\in\alpha$ interchangeably.

\paragraph{Graphs}
We use standard graph notation (see, e.g., \cite{Diestel25}).
All graphs considered in~this paper are simple (i.e., no loops or multiedges)
and undirected. We write $V(G)$ and~$E(G)$ for the~vertex and
edge set of~a~graph $G$, respectively.
We use $vw$ as a~shorthand for an~edge $\{v,w\}$ in~$G$.
For~$v \in V(G)$, we write $N_G(v) \coloneqq \{w \in V(G)\mid vw \in E(G)\}$
to denote the~neighbourhood of~$v$. Let $G,H$ be two graphs with
$|V(G)| = |V(H)| = n$, and~let $\pi$ be a~bijection from $V(G)$ to~$V(H)$.
For~a~set of~2-element subsets $P \subseteq \binom{V(G)}{2}$,
we write $P^\pi \coloneqq \{ \{\pi(v),\pi(w)\}\st \{v,w\}\in P\}$.
The~\emph{edit cost of~$\pi$ with respect to~$G,H$}, denoted
by~$\costGED_{G,H}(\pi)$, is the~cardinality of~the~set
$E(G)^\pi \symdiff E(H)$, i.e., the~sum of~the~number of
edges of~$G$ sent by $\pi$ to~non-edges of~$H$
and~the~number of~non-edges of~$G$ sent by~$\pi$ to~edges of~$H$.
We usually write $\costGED(\pi)$ if the~graphs $G,H$ are clear from context.
The \emph{graph edit distance from $G$ to~$H$}, denoted by~$\ged(G,H)$,
is the~minimum edit cost amongst all bijections between $V(G)$ and~$V(H)$.

An \emph{edge-weighted} graph is a~tuple $(G,\weight_G)$ where
$\weight_G\colon E(G) \to \mathbb R$ is~a~weight function.
Let \mbox{$(H,\weight_H)$} be a~second edge-weighted graph.
For a~bijection $\pi\colon V(G) \to V(H)$, we define
$\costGED_{G,H}(\pi)
    \coloneqq \sum_{v,w \in V(G)} |\weight_G(vw) - \weight_H(\pi(v)\pi(w))|$
(as a~convention, we define $\weight(vw) \coloneqq 0$ if~$vw \notin E(G)$).
Then the~graph edit distance is defined as~before.

A \emph{coloured} graph is a~tuple $(G,\lambda_G)$ with
$\lambda_G\colon V(G) \to C$, where $C$ is a~finite set of~colours.
The sets $\lambda_G^{-1}(c)$, $c \in C$, are the~\emph{colour classes}
of~$\lambda_G$. The graph edit distance over coloured graphs is defined
by~considering only those bijections that preserve the~colour of~the~vertices.

\paragraph{Quadratic Assignment Problems.}
Recall the~definition of~the~quadratic assignment problem (QAP) from the
introduction. We call the~domain size $n$ the~\emph{order} of~an instance.
Also recall that a~QAP instance is $B$-bounded if the~absolute value of~all
coefficients is at most $B$. We write $\mathcal A$ to~denote integral feasible
solutions $x$ (to the~constraints \eqref{eq:qap1}, \eqref{eq:qap2}); each such
solution corresponds to~a~bijection $\phi_x\colon[n]\to[n]$ mapping every
$v \in [n]$ to~the~unique $v' \in [n]$ such that $x(v,v') = 1$. In this case,
we have
\[ \cost(x) = \sum_{v\in[n]}\sum_{w\in[n]} c(v,\varphi_x(v),w,\varphi_x(w)).\]

\begin{example}[\cite{AroraFK02}]
  \label[example]{ex:ged_qap}
  Let $G,H$ be two graphs with $V(G) = V(H) = [n]$.
  Consider the~coefficient function
  \begin{equation}
    \label{eq:1}
        c(v,v',w,w') \coloneqq\begin{cases}
      1&\text{if }(v,w)\in E(G)\iff (v',w')\not\in E(H)\\
      0&\text{otherwise}
    \end{cases}
  \end{equation}
  An optimal integral solution to~the~QAP defined by~$c$ is an
  assignment $x^*$ such that the~bijection $\varphi_{x^*}: V(G)\to
  V(H)$ has minimal edit cost. As such, $\cost(x^*) = 2 \cdot \ged(G,H)$.
  The factor 2 is necessary because in~the~QAP cost, each edge is counted twice.
  Similarly, we can formulate a~QAP for weighted graphs $(G,\omega_G)$
  and~$(H,\omega_H)$ by~letting
  \begin{equation}
    \label{eq:graph-qap-cost-function}
     c(v,v',w,w')\coloneqq|\omega_G(v,w)-\omega_H(v',w')|.
  \end{equation}
\end{example}

\section{VC~Dimension}
\label{sec:vc_dimension}

In this section, we introduce and compare various set systems for (weighted) graphs
and QAP instances, which give rise to corresponding notions of~VC~dimension.
We start by recalling basic definitions and~standard tools in~this~context.

\paragraph{Basics.}
A \emph{set system on V} is a~set $\mathcal H\subseteq 2^V$
(where $2^V$ denotes the~power set of~$V$).
A subset $X\subseteq V$ is \emph{shattered by $\mathcal H$} if
$\mathcal H\cap X \coloneqq \{ H\cap X \mid H\in\mathcal H\} = 2^X$.
The \emph{VC~dimension} $\VC(\mathcal H)$ of~$\mathcal H$ is
the~maximal cardinality of~a subset $X\subseteq V$ shattered by $\mathcal H$.
Let $\varepsilon > 0$. An \emph{$\varepsilon$-net} for $\mathcal H$
is a~set $S\subseteq V$ such that for all $H\in\mathcal H$
with $|H| > \varepsilon\cdot |V|$, we have $H\cap S \ne\emptyset$.
Furthermore, we say that $S$ is an~\textit{$\epsilon$-approximation
for $\mathcal H$} if for all
$H \in \mathcal H$,
\[ |H| \in \left[ \frac{n}{|S|} |H \cap S|  \pm \epsilon n\right].\]
Let us recall a~few combinatorial lemmas concerning VC~dimension.

\begin{lemma}[\cite{Matousek99}]
  \label[lemma]{lem:epsilon-net-general}
  For every set system $\mathcal H$ and $\epsilon > 0$,
  there is an~$\epsilon$-net $S$ for $\mathcal H$ with
  $|S| \leq \frac{\ln |\mathcal H|}{\epsilon}$.
\end{lemma}
\begin{lemma}[\cite{haussler_e-nets_1987}]
  \label[lemma]{lem:epsilon-net}
  For every set system $\mathcal H$ of~VC~dimension $d$
  and every $\epsilon > 0$, there is an~$\epsilon$-net
  $S$ for $\mathcal H$ of~size $|S| \leq O\left(\frac{d}{\epsilon} \log \frac{1}{\epsilon} \right).$
\end{lemma}

\begin{lemma}[\protect{\cite{li_improved_2001}, \cite[Theorem 47.2.2]{goodman2025handbook}}]
  \label[lemma]{lem:epsilon-approximation-whp}
  For every set system $\mathcal H$ of~VC~dimension $d$
  and~every $\epsilon, \gamma > 0$, a~uniform sample of~size
  $O\left( \frac{1}{\epsilon^2} \left( d + \log\frac{1}{\gamma} \right) \right)$
  is an~$\epsilon$-approximation with probability at~least $1-\gamma$.
\end{lemma}

\paragraph{Graphs.}
Next, let us define the~VC~dimension of~(weighted) graphs.
Let $G$ be a~graph.
The~\emph{neighbourhood set system} of~$G$ is defined by
\begin{equation}
  \label{eqn:neighbourhood_set_system}
   \mathcal N_G \coloneqq \{ N_G(v)\mid v\in V(G)\}.
\end{equation}
We define the~\emph{VC~dimension} of~$G$ to be the~VC~dimension
of $\mathcal N_G$. Let us remark at this point that various possible
definitions for the~VC~dimension of~a graph have been considered
in~the~literature by associating different set systems (see,
e.g., \cite{ChanCGKLZ25,DucoffeHV22,DurajKP24,LeW24}).
However, those set systems are usually supersets of~$\mathcal N_G$,
i.e., our definition of~VC~dimension is the~most relaxed (and thus,
our algorithmic results are applicable more widely).

Let $(G, \weight)$ be an~edge-weighted graph with weight function
$\weight\colon E(G) \to \mathbb R$ (recall that we define
$\weight(v,w) \coloneqq 0$ if $vw\not\in E(G)$ as a~convention).
We consider a~family of~\emph{threshold set systems} over $V(G)$ as follows.
For each $t \in \mathbb R$, let $N^t_G(v) \coloneqq \{ w\in V(G)\mid
\weight(v,w) > t\}$, and~let~$\mathcal N^t_G \coloneqq \{N^t_G(v)\mid
v\in V(G)\}$.
This is simply the~neighbourhood set system for
the~\emph{threshold graph} $G^t$, which is the~graph with the~same
vertex set as $G$ and~edge set consisting of~all edges of~weight $>t$.
We define the~VC~dimension of~$(G,\weight)$ to be the~maximum VC
dimension of~$\mathcal N_G^t$ over all $t \in \mathbb R$.
Observe that the~maximum exists because the~threshold graphs form
an~increasing sequence of~subgraphs of~$G$, which implies that there are
at most $|E(G)|+1$ distinct threshold graphs. 
Furthermore, if all edge-weights are equal to a~fixed constant
$a \neq 0$, then the~VC~dimension of~$(G,\weight)$ is identical to the~VC
dimension of~the~unweighted graph $G$.

\paragraph{QAPs.}
For QAPs, we can define the~VC~dimension in~a~similar spirit.
Recall that a~QAP instance of~order $n$ is given by the
coefficients $c(v,v',w,w')$ for all $v,v',w,w' \in [n]$.
The basic idea is to view the~coefficient function as an
edge-weighted graph over vertex set $[n] \times [n]$
(i.e.,~the~weight of~an edge between $(v,v')$ and~$(w,w')$ is
$\weight((v,v'),(w,w')) \coloneqq c(v,v',w,w')$).
With this in~mind, we again define a~threshold set system
for every $t \in \mathbb R$ similar to above as
\begin{align}
  \label{eqn:def-of-htc}
  \begin{split}
  \mathcal H^t_{c} &\coloneqq \{ M^t_c(v,v')\mid (v,v')\in [n]\times[n]\},\\
  \text{where } M^t_c(v,v') &\coloneqq \{(w,w')\mid c(v,v',w,w') > t\}.
  \end{split}
\end{align}
In other words, each $(v,v')$ induces a~hypothesis, made of~exactly
those associations $(w,w')$ which, taken together with $(v,v')$,
would not incur a~cost above our fixed threshold.

Let us point out that there is a~second, more relaxed definition
for the~VC~dimension of~a~QAP instance that is actually sufficient
for all our algorithmic applications. Indeed, while the~definition
above does capture the~combinatorial complexity of~the~subsets
of $[n]\times [n]$ definable from the~coefficient function
$c(v,v',w,w')$, it does not account
for one additional feature of~the~problem at hand: a~solution to
a QAP is not an~arbitrary subset of~$[n]\times [n]$, but a~perfect
matching, i.e., the~graph of~a bijection $\varphi\colon [n]\to[n]$.
Yet, the~set systems $\mathcal H^t_c$ contain many sets which are not subsets
of such a~matching. As such, there may be some large subsets of~$[n]\times [n]$
which are irrelevant to the~complexity of~the~solution space of~the~QAP at hand,
while being shattered by $\mathcal H^t_c$. In~contrast, for a~fixed bijection
$\varphi : [n]\to[n]$, let
\begin{align}
\begin{split}
  \label{eqn:def-of-htcphi}
  \mathcal H^t_{c,\phi}
        &\coloneqq \{ M^t_{c,\phi}(v,v')\mid (v,v')\in [n]\times[n]\},\\
  \text{where } M^t_{c,\phi}(v,v')
        &\coloneqq \{(w,\phi(w))\mid c(v,v',w,\phi(w)) > t\}.
\end{split}
\end{align}
Note that $\mathcal H^t_{c,\varphi}$ is the~restriction of~$\mathcal
H^t_c$ to $\graph(\varphi) := \{(v,\varphi(v))\mid v\in[n]\}$. As
such,
\begin{equation}
  \label{eqn:htcphi_le_htc}
\VC(\mathcal H^t_{c,\phi}) \le \VC(\mathcal H^t_c).
\end{equation}
(We omit the~index $c$ in~$M^t_c,\mathcal H^t_c,M^t_{c,\phi}$ and
$\mathcal H^t_{c, \phi}$ if the
coefficient function $c$ is clear from  context.)

We say that a~QAP represented by the~coefficient function $c$
has a~\emph{VC~dimension weakly bounded by $d$}, if for all $t\in\mathbb R$
and all bijections $\phi:[n] \to [n]$ it holds that
$\VC\left(\mathcal H^t_{c,\phi}\right) \leq d$.

\begin{remark}
  Having VC~dimension weakly bounded by $d$
  is testable in~time $n^{O(d)}$. Namely, for each $t \in \mathbb R$, 
  the~set system $\mathcal H^t_{c,\varphi}$ has VC~dimension at most $d$
  for any $\varphi$ if and~only if for all injective partial function
  $\alpha: [n]\pfun [n]$ with $|\alpha| = d+1$ there is a~restriction
  $\alpha' \subseteq \alpha$ such that there are no $v,v'$ with
  $\alpha' = \{(w,\alpha(w)) \in \alpha \mid t < c(v,v',w,\alpha(w)) \}$.
  Because the~function $c$ takes at most $n^4$ different values, we only
  need to consider $n^4 + 1$ different values of~$t$.
\end{remark}

Note that by \eqref{eqn:htcphi_le_htc}, if a~QAP has VC~dimension
bounded by $d$, then it has VC~dimension weakly bounded by $d$.
However, the~other direction is not true in~general.

\begin{lemma}
  \label[lemma]{lem:htcphi_lt_htc}
  For each $n\geq 4$, there is a~cost function
  $c: [n]\times [n] \times [n]\times [n] \to \{0,1\}$,
  such that $\VC(\mathcal H^{0}_{c}) = \floor{\log n}$, whereas for each
  bijection $\phi$ we have $\VC(\mathcal H^{0}_{c,\phi}) \leq 1$.
\end{lemma}
\begin{proof}
  Let $k = \floor{\log n}$. We define the~cost function $c$ as follows.
  Let $v_{A} \in [n]$ be distinct vertices for all sets $A \subseteq [k]$;
  such vertices exist, since there are $2^k \leq n$ subsets of~$[k]$. Define
  \[ c(n,v_A,1,x) \coloneqq \begin{cases}
    1 & \text{if }x \in A\\
    0 & \text{otherwise,}
  \end{cases}\]
  and~all other undefined values are $0$.
  Then the~VC~dimension of~$\mathcal H^{0}_{c}$ is $\floor{\log n}$,
  since the~set $\{(1, x) \mid x \in [k]\}$ is shattered.

  Furthermore, let $\phi:[n] \to [n]$ be an~arbitrary bijection.
  Note that each set $M^0_{c,\phi}(v,v')$ contains at most one element,
  namely $(1,\phi(1))$. This means that any set of~size $\geq 2$
  cannot be shattered.
\end{proof}

\paragraph{Comparison.}

One of~main problems considered in~this paper is the~GED, which can be naturally
translated into a~QAP, as seen in~\cref{ex:ged_qap}.
Towards this end, we show that two weighted graphs of~bounded VC
dimension translate to a~QAP instance of~bounded VC~dimension (in~fact,
the~VC~dimension is linearly bounded).

The following lemma can be proved by standard arguments (see, e.g., \cite{Matousek99}).
  
\begin{lemma}
  \label[lemma]{lem:general-10d-lemma}
  Let $G = (V,E, \weight_G), H = (V, E, \weight_H)$ be edge-weighted graphs
  of~VC~dimensions $d_G, d_H \geq 1$, respectively,
  and~let $d = \max\{d_G,d_H\}$.
  Let $c$ be the~coefficient function defined
  in~\eqref{eq:graph-qap-cost-function}.
  Then $\VC(\mathcal H^t_{c})\leq 10d$ for every $t \in \mathbb R$.
\end{lemma}

Towards the~proof, we will need the~following standard result about
VC~dimension, known as the~Sauer-Shelah lemma.

\begin{lemma}[\cite{Sauer72}]
  \label[lemma]{Sauer-Shelah}
  Let $\mathcal H$ be a~set system on $V$ of~VC~dimension $d$. Then
  \[ \max_{X\subseteq V, |X| = s} |\{H \cap X \mid H \in \mathcal H\}|
          \leq \left(\frac{es}{d}\right)^d\]
  for all $s \geq d$.
\end{lemma}

\begin{proof}[Proof of~\cref{lem:general-10d-lemma}]
  Suppose a~set $X \subseteq V(G) \times V(H)$ of~size $s$ shattered by
  $\mathcal H^t_c$. Note that for each $v_1, v_2 \in V(G)$ and~each
  $v'\in V(H)$ such that $M^t(v_1, v') = M^t(v_2,v')$, then
  $N^t_G(v_1) = N^t_G(v_2)$. Therefore, $X \cap M^t(v,w)$ only depends on
  $\pi_1(X) \cap N^t_G(v)$ and~$\pi_2(X) \cap N^t_H(w)$.
  (Here, $\pi_i(X)$ is the~projection onto the~$i$-th coordinate for $i = 1,2$.)
  \begin{align*}
    2^s &= \left| \left\{
                  X \cap M^t(v,w) \mid v\in V(G), w \in V(H) \right\} \right|\\
    &\leq \left| \left\{\pi_1(X) \cap N^t(v) \mid v\in V(G) \right\} \right|
          \cdot \left| \left\{
                    \pi_2(X) \cap N^t_H(w) \mid w\in V(H) \right\} \right|\\
    &\leq \left(\frac{es}{d}\right)^{2d}.
  \end{align*}
  Let $c \coloneqq \frac{s}{d}$. Then we have $2^{c} \leq (ec)^2$, which implies
  $c < 10$. Therefore, any shattered set has size less than $10d$.
\end{proof}

For unweighted graphs, we can prove that all above mentioned notions
of VC~dimensions are linearly dependent.
\begin{lemma}
  \label{lem:Htc-lower-bound}
  Let $G, H$ be unweighted graphs
  of~VC~dimensions $d_G, d_H \geq 1$, respectively and~let
  $d = \max\{d_G,d_H\}$.
  Let $c$ be the~coefficient function defined in~\eqref{eq:1}
  and~let $\phi: V(G)\to V(H)$ be an~arbitrary bijection.
  Then $d \leq \VC(\mathcal H^0_{c,\phi}) \leq \VC(\mathcal H^0_{c})\leq 10d$.
\end{lemma}
\begin{proof}
  The second inequality holds by \eqref{eqn:htcphi_le_htc}, and~the~third
  inequality comes from \cref{lem:general-10d-lemma}. For the first inequality,
  let $\phi: V \to W$ be a~bijection.
  Without loss of~generality, suppose that $d_G \geq d_H$. (Otherwise, swap
  the~graphs through the~bijection $\phi$.) Let
  $\bar{\phi} : 2^V\to 2^{V\times W}$ defined by
  $\bar\phi(X) := \{ (x,\phi(x))\mid x\in X\}$. Note that $\bar\varphi$ is
  a~bijection between subsets of~$V$ and~subsets of~$\bar\phi(V)$.
  Let $X\subseteq V$ be shattered of~size $d$ by $\mathcal N_G$.
  We claim that $\bar X \coloneqq \bar\phi(X)$ is shattered by
  $\mathcal M^0_\phi$. Note that
  $M^0_\phi(v,w) = \bar\phi(N_G(v)) \symdiff \bar\phi(\phi^{-1}(N_H(w)))$.
  Let $\bar X' \subseteq \bar X$. Let $w \in W$ be arbitrary.

  There is $X'\subseteq X$ such that $\bar\varphi(X') = \bar X'$.
  Since $X$ is shattered by $\mathcal N_G$, we have a~vertex $v\in V$
  such that $N_G(v) \cap X = X' \symdiff \left(\phi^{-1}(N_H(w)) \cap X\right)$.
  Then
  \begin{align*}
      M^0_\phi(v,w) \cap \bar X &=
          \left(\bar\phi(N_G(v)) \symdiff \bar\phi(\phi^{-1}(N_H(w)))\right)
                \cap \bar X\\
      &= \left(\bar\phi(N_G(v)) \cap \bar X\right)
            \symdiff \left(\bar\phi(\phi^{-1}(N_H(w))) \cap \bar X\right)\\
      &= \bar\phi\left(N_G(v) \cap X\right)
              \symdiff \left(\bar\phi(\phi^{-1}(N_H(w))) \cap \bar X\right)\\
      &= \bar X' \symdiff \left(\bar\phi(\phi^{-1}(N_H(w))) \cap X\right)
                 \symdiff \left(\bar\phi(\phi^{-1}(N_H(w))) \cap X\right)\\
      &= \bar X'
  \end{align*}
  Hence, $d \leq \VC(\mathcal H^0_{c,\phi})$.
\end{proof}

\section{Approximation with bounded VC~dimension}
\label{sec:approximation-with-bounded-VC-dimension}
As mentioned in~the~introduction, Arora et al.~\cite{AroraFK02} show
that optimal solutions to~$\Bbound$-bounded QAPs can be approximated
in quasipolynomial time. As a~result, they obtain
a quasipolynomial time approximation for GED. Building
on this work, we show in~this section that under the~additional
assumption that the~QAP instance under consideration has size $n$ and
VC~dimension at~most $d$, its optimal solution can be approximated
up to~an~additive error of~$\varepsilon n^2$ in~time
$n^{O\left(\Bbound^2\epsilon^{-2}(d + \log \Bbound\epsilon^{-1})\right)}$
(\cref{thm:quadratic-assignment-bounded-vc-approximation}).
Because the~reduction from the~GED problem to~QAP preserves VC~dimension
(see \cref{lem:general-10d-lemma}), this result yields an~algorithm
approximating GED up to~an~additive error of~$\varepsilon n^2$ for graphs
of VC~dimension bounded by $d$, which runs in~time $n^{O(\varepsilon^{-2}d)}$
(proving \cref{thm:main-edit-distance}).

Let us rephrase the~quasipolynomial time approximation algorithm for QAPs
provided in~\cite{AroraFK02} in~our setting. For any injective partial function
$\alpha: [n] \to [n]$, define
\begin{align*}
  b_{\alpha}(v, v') &\coloneqq
                      \frac{n}{|\alpha|} \sum_{(w,w') \in \alpha} c(v,v',w,w').
\end{align*}
Although not stated explicitly, \cite[Section 3]{AroraFK02} is devoted to
the proof of~the~following lemma.
\begin{lemma}
  \label[lemma]{arora-lemma}
    There is an~algorithm that, given any $\Bbound$-bounded QAP and
    $\epsilon > 0$, finds a~perfect matching $x$ such that
    $\cost(x) \leq \cost(x^*) + \epsilon n^2$, where $x^*$ is a~fixed optimal
    solution to~the~input QAP. The~algorithm runs in~time $n^{O(m)}$,
    where $m$ is the~smallest integer such that there is
    an~$\alpha \subseteq \graph(\phi_{x^*})$ of~size $m$, such that
    for all $v,v'\in [n]$
    \begin{equation}
      \label{eq:approximation-of-b}
      b_{\phi_{x^*}}(v,v') \in
          \left[b_{\alpha}(v,v') \pm \frac13 \epsilon n\right].
    \end{equation}
\end{lemma}
For completeness, we provide a~proof of~\cref{arora-lemma}
in~\cref{app:assignment-problem}. Arora et al.\ obtain a~quasipolynomial time
additive approximation algorithm for QAP by additionally proving the~existence
of~such $\alpha$ of~size $m = O\left(\Bbound^2\epsilon^{-2}\log n\right)$.
In~the~following, we show that $m$ can be bounded by the~VC~dimension
of~the~QAP at~hand, yielding a polynomial-time algorithm for bounded
VC~dimension QAPs.

Let $x^*$ be a~fixed optimal integral solution to~the~input QAP,
and let $\phi^* \coloneqq \phi_{x^*}$ be the~bijection $[n]\to [n]$
corresponding to~$x^*$. Let $\mathcal H^t_{\phi^*}$ be the~hypergraph
for $\phi^*$ as defined in~\eqref{eqn:def-of-htcphi}.

In order to~ease exposition, let us first consider the~restricted case where
the coefficient function $c$ is $\{0,1\}$ valued.
Note that for any $\alpha \subseteq \graph(\phi^*)$, we have
\begin{equation}
  \label{eq:b-same-as-M}
  b_\alpha(v,v') = \frac{n}{|\alpha|} \sum_{(w,w') \in \alpha} c(v,v',w,w')
                 = \frac{n}{|\alpha|} |M^0_\alpha(v,v')|.
\end{equation}
where $M^t_\alpha(v,v') := M^t_{\varphi^*}(v,v')\cap \alpha$.
Now, by \cref{lem:epsilon-approximation-whp} there is an~$\alpha$
of size $O\left(d/\epsilon^2\right)$ such that
\[ b_{\phi^*}(v,v') = |M^0_{\phi^*}(v,v')| \in
    \left[ \frac{n}{|\alpha|}|M^0_\alpha(v,v')| \pm \frac{\epsilon}{3} n \right]
    = \left[ b_\alpha(v,v') \pm \frac{\epsilon}{3} n \right] \]
which concludes the~argument.

In order to~extend this proof to~any $\Bbound$-bounded QAP of~VC~dimension
$d$, we discretise $[-\Bbound,\Bbound]$ using a~set $T$ of~thresholds which
enable a~suitable approximation of~$b_{\phi^*}$. Divide the~interval
$[-\Bbound, \Bbound]$ into $k \coloneqq \ceil{\frac{24\Bbound}{\epsilon}}$
intervals of~length $\frac{2\Bbound}{k}$. Let
$T = \{ -\Bbound, -\Bbound + \frac{2\Bbound}k, \ldots,
    \Bbound - \frac{2\Bbound}k\}$
be the~set of~(left) boundaries of~these intervals.

The following lemma provides a~statement similar to~\eqref{eq:b-same-as-M}.
While in~the~general case we do not obtain equality, we show that averaging
$|M^t_\alpha(v,v')|$ over $t\in T$ yields a~good approximation
of $b_\alpha(v,v')$.

\begin{lemma}
  \label[lemma]{lem:approx-b-alpha-by-mean}
 For any $\alpha \subseteq \graph(\phi^*)$ and any $v,v' \in [n]$, we have
 \[  b_{\alpha}(v,v') \in \left[ \frac{2\Bbound}{k}
          \sum_{t\in T} \frac{n}{|\alpha|} \left|M^{t}_{\alpha}(v,v')\right|
              - \Bbound n \pm \frac{2\Bbound}{k} n \right].\]
\end{lemma}
\begin{proof}
  Let $v,v'\in[n]$. For each $t\in T$ and $(w,\alpha(w)) \in \alpha$, we define
  the~``bit''
  \[ \beta(w,t) \coloneqq \begin{cases}
    1 & \text{if $t < c(v,v',w,\alpha(w))$,}\\
    0 & \text{otherwise.} \end{cases}\]
  By definition, we have
  \[ \left| M^{t}_{\alpha}(v,v') \right| =
            \sum_{(w,\alpha(w)) \in \alpha} \beta(w,t),\]
  and for fixed $w\in [n]$,
  \[ \sum_{t\in T} \beta(w,t) =
            \floor{\frac{k \cdot c(v,v',w,\alpha(w)) + k\Bbound}{2\Bbound}}
            \in \left[\frac{k \cdot c(v,v',w,\alpha(w))}{2\Bbound} + \frac{k}{2}
                    \pm 1\right].\]
  Then we have
  \begin{align*}
      \frac{1}{k}\sum_{t\in T} \frac{n}{|\alpha|}
        \left|M^{t}_{\alpha}(v,v')\right| &=
          \frac{n}{k|\alpha|}\sum_{t\in T} \sum_{(w,\alpha(w))\in \alpha}
                                                                  \beta(w,t)\\
      &= \frac{n}{k|\alpha|}\sum_{(w,\alpha(w))\in \alpha} \sum_{t\in T}
                                                                  \beta(w,t)\\
      &\in \left[\frac{n}{k|\alpha|}\sum_{(w,\alpha(w))\in \alpha}
                \left(\frac{k \cdot c(v,v',w,\alpha(w))}{2\Bbound} + \frac{k}{2}
                      \pm 1\right)\right]\\
      &= \left[\frac{1}{2\Bbound}b_{\alpha}(v, v') + \frac{n}{2}
               \pm \frac{1}{k} n\right].
    \end{align*}
    Rearranging concludes the~proof of~the~lemma.
\end{proof}
 
Let $K \coloneqq \min\{ k, |\{c(v,v',w,w') \mid v,v',w,w'\in [n]\}|\}$.
We now provide an~$\left(\frac{\epsilon}{12B}\right)$-approx\-imation of~the
set systems $\mathcal H^t_{\phi^*}$ for all $t$. Intuitively, we use a~union
bound over all values $t$ on the~probability that a~sample is not
$\left(\frac{\epsilon}{12B}\right)$-approximation and get the~following result. 
\begin{lemma}
  \label[lemma]{lem:approximation-existence-for-all-t}
  There is an~$\alpha$ of~size $O\left(\frac{B^2}{\epsilon^2}(d + \ln K)\right)$
  such that for all $t \in T$ and all $v,v'$,
  \[ \left|M^{t}_{\phi^*}(v,v')\right| \in
          \left[ \frac{n}{|S|} \left|M^{t}_{\phi^*}(v,v') \cap \alpha\right|
              \pm \frac{\epsilon}{12B} n\right]. \]
\end{lemma}
\begin{proof}
  Observe first that the~set system $\mathcal H^t_{\phi^*}$ only changes when
  $t$ passes a~value of~the~form $c(v,v',w,w')$ for some $v,v',w,w' \in [n]$.
  Hence, the~family $\{\mathcal H^t_{\phi^*} \mid t \in T\}$ contains at~most
  $K$ distinct set systems.
  For fixed $t$ inducing a~distinct set system, we have by
  \cref{lem:epsilon-approximation-whp} that a~uniform sample $\alpha$ of~size
  $O\left(\frac{c^2}{\epsilon^2}\left(d + \log\frac{1}{\gamma} \right)\right)$
  is an~$\frac{\epsilon}{12\Bbound}$-approximation for $\mathcal H_{\phi^*}^{t}$
  with probability at~least $1-\gamma$.
  By the~union bound on the~complement event and setting
  $\gamma \coloneqq \frac1{2K}$, we have that the~probability
  of~a~random sample $\alpha$ of~size
  $O\left(\frac{c^2}{\epsilon^2}(d + \ln K)\right)$
  being an~$\frac{\epsilon}{12\Bbound}$-approx\-imation for
  $\mathcal H_{\phi^*}^{t}$ for all $t \in T$ is positive, in~particular, such
  set exists.
\end{proof}

Similarly to~the~0-1-valued case, the~following result follows from
\cref{lem:approx-b-alpha-by-mean} and
\cref{lem:approximation-existence-for-all-t}.
\begin{lemma}
  \label[lemma]{cor:approximation-existence-for-bs}
  There is an~$\alpha$ of~size
  $O\left( \frac{\Bbound^2}{\epsilon^2}(d + \log K)\right)$ such that for all
  $v, v'$,
  \[ b_{\phi^*}(v,v') \in
        \left[ b_{\alpha}(v,v') \pm \frac{\epsilon}{3}n\right]\]
\end{lemma}
\begin{proof}
  By \cref{lem:approximation-existence-for-all-t} we get an $\alpha$
  of~the~given size such that
  \begin{equation}
    \label{eq:half-aprox}
      \left|M^{t}_{\phi^*}(v,v')\right| \in
          \left[ \frac{n}{|\alpha|} \left|M^{t}_{\alpha}(v,v') \right|
                \pm \frac{\epsilon}{12\Bbound} n \right]
  \end{equation}
  for all $t\in T$.
  Then
  \begin{align*}
    b_{\phi^*}(v,v') &\in \left[ \frac{2\Bbound}{k}
        \sum_{t \in T} \left|M^{t}_{\phi^*}(v,v')\right| - \Bbound n
        \pm \frac{2\Bbound}{k} n \right]
            & \text{\cref{lem:approx-b-alpha-by-mean}}\\
    &\subseteq \left[\frac{2\Bbound}k
        \sum_{t \in T}\left(\frac{n}{|\alpha|} \left|M^{t}_{\alpha}(v,v')\right|
        \pm \frac{\epsilon}{12\Bbound} n\right) - \Bbound n
          \pm \frac{2\Bbound}k n \right]
            & \eqref{eq:half-aprox}\\
    &= \left[\frac{2\Bbound}k \sum_{t \in T} \frac{n}{|\alpha|}
        \left|M^{t}_{\alpha}(v,v')\right| - \Bbound n
        \pm \left(\frac{2\Bbound}k n + \frac{\epsilon}{6} n \right)\right]
            & \\
    &\subseteq \left[ \left(b_{\alpha}(v,v') \pm \frac{2\Bbound}{k}n \right)
        \pm \left(\frac{2\Bbound}{k} n + \frac{\epsilon}{6} n \right)\right]
            & \text{\cref{lem:approx-b-alpha-by-mean}}\\
    &= \left[ b_{\alpha}(v,v')
        \pm \left(\frac{4\Bbound}{k}n + \frac{\epsilon}{6} n\right)\right]
            & \\
    &\subseteq \left[ b_{\alpha}(v,v') \pm \frac{\epsilon}{3} n\right],
            & \text{since $k = \ceil{\frac{24\Bbound}{\epsilon}}$.}\\
  \end{align*}
\end{proof}
Hence, \cref{cor:approximation-existence-for-bs} guarantees the~existence of~an
$\alpha$ of size $m = O\left(\frac{\Bbound^2}{\epsilon^2}(d + \log K)\right)$
satisfying the~requirements of~\cref{arora-lemma}, implying that we
can find a~perfect matching $x$ with
\mbox{$\cost(x) \leq \cost(x^*) + \epsilon n^2$} in time $n^{O(m)}$.
Since $K \leq k = \ceil{\frac{12\Bbound}{\varepsilon}}$, we obtain the following
result.
\begin{theorem}
   \label{thm:quadratic-assignment-bounded-vc-approximation}
    There is an~algorithm that, given $\varepsilon > 0$ and a~$\Bbound$-bounded
    QAP instance of~VC~dimension weakly bounded by $d$, outputs a~perfect
    matching $x$ such that \mbox{$\cost(x) \leq \cost(x^*) + \epsilon n^2$}
    in~time
    $n^{O\left(\Bbound^2\epsilon^{-2}(d + \log \Bbound\epsilon^{-1})\right)}$,
    where $x^*$ denotes an~optimum solution.
\end{theorem}

In particular, we obtain
\cref{thm:quadratic-assignment-bounded-vc-approximation-main} via
\eqref{eqn:htcphi_le_htc}. In~the~special case where the~coefficient function
takes a~moderate number of~different values, we have the following corollary.

\begin{corollary}
  \label[corollary]{cor:qap_vc_approx_bounded_values}
  There is an~algorithm that, given $\varepsilon > 0$ and a~$\Bbound$-bounded
  QAP instance of~VC~dimension weakly bounded by $d$, outputs a~perfect matching
  $x$ such that \mbox{$\cost(x) \leq \cost(x^*) + \varepsilon n^2$}, where $x^*$
  denotes an~optimum solution. If the~coefficient function of~the~QAP instance
  takes $\kappa$~different values, the~algorithm runs in~time
  $n^{O(\Bbound^2\varepsilon^{-2}(d + \log \kappa))}$.
\end{corollary}

Let us now turn the~problem GED of~computing the~edit distance between two
graphs. The~objective is to~find a~bijection $\phi$ between the~vertex
sets of~the~two input graphs that minimises the~cost, defined as the~number
of edge mismatches. We have seen in~\cref{ex:ged_qap} how to~translate this
into a~QAP. Furthermore, by \cref{lem:general-10d-lemma}, if the~VC~dimension
of the~two input graphs is bounded by $d$ then the~VC~dimension of~the~QAP is
weakly bounded by $O(d)$. Thus, as a~corollary to
\cref{thm:quadratic-assignment-bounded-vc-approximation}, we obtain
the following result.

\begin{corollary}\label{cor:ged-weighted}
  There is an~algorithm that, given $\epsilon > 0$ and weighted
  $n$-vertex graphs $G,H$ of~VC~dimension at most $d$,
  outputs a~bijection $\phi:V(G)\to V(H)$
  such that \mbox{$\cost(\phi) \leq \ged + \epsilon n^2$} in~time
  $n^{O\left(B^2\epsilon^{-2}(d + \log B\epsilon^{-1})\right)}$,
  where $B$ bounds the edge weights (in absolute value) in $G$ and $H$.
\end{corollary}

For unweighted graphs, GED translates to~a~0-1-valued QAP, so we can
apply  \cref{cor:qap_vc_approx_bounded_values} and obtain the~slightly
better bound stated as \cref{thm:main-edit-distance} in~the~introduction.

\begin{remark}\label{rem:higher-order}
  It is also possible to~generalise our proof for higher-degree assignment
  problems. Similar to~\cite[Section 3]{AroraFK02}, we can proceed by induction
  on the~degree $r$. For $r > 2$, we have to~find $\alpha$ such that
  $b_{\phi^*}(v_2, v_2', \ldots, v_r,v_r') \in
      \left[b'_{\alpha}(v_2, v_2', \ldots, v_r,v_r')
            \pm \frac{1}{r}\epsilon n\right]$.
  We can define the~set systems accordingly, so that the
  $\epsilon$-approximations will guarantee the~existence of~such small sample
  $\alpha$. Furthermore, the~average will again approximate the~$b$-values.
  An~analogous variant of~\Cref{cor:approximation-existence-for-bs}
  will then be used for $\epsilon/r$ in~place of~$\epsilon/3$. Then we can show
  an~algorithm running in~time
  $n^{O\left(\Bbound^2r^2\epsilon^{-2}(d + \log \Bbound r\epsilon^{-1})\right)}$
  that additively approximates the~solution (up to~$\epsilon n^r$) for every
  $\Bbound$-bounded degree-$r$ assignment problem. Since $d \leq \log n^{2r}$,
  this slightly improves the~corresponding result
  in~\cite[Section 3]{AroraFK02}.
\end{remark}

\section{Robust Graph Isomorphism and Weisfeiler-Leman Algorithm}
\label{sec:wl}

A particularly interesting special case of~GED is the
\emph{robust graph isomorphism} problem. This is a~promise problem, where we are
given two graphs $G,H$, which are promised to be either isomorphic or have large
edit distance, and the~task is to decide which of~the~two options holds.
Formally, for a~fixed $\epsilon > 0$, we consider the~promise problem
$\epsilon$-GI defined as follows:

\medskip
\defpromiseproblem{$\epsilon$-\textsc{GI}}
    {Two $n$-vertex graphs $G,H$.}
    {Either $G \cong H$ or $\ged(G,H) \geq \epsilon n^2$.}
    {Is $G\cong H$?}
\medskip

By \cref{thm:main-edit-distance} the~problem $\epsilon$-\textsc{GI} can
be solved in~time $n^{O(\epsilon^{-2}d)}$ on graphs of~VC~dimension at most $d$.
Still, the~algorithm is fairly complicated relying on, e.g., advanced rounding
methods in~the~proof of~\cref{arora-lemma}.

In this section, we show that, for the~special case of~$\epsilon$-\textsc{GI},
we can instead use the~Weisfeiler-Leman algorithm. For $k \geq 1$, the
\emph{$k$-dimensional Weisfeiler-Leman algorithm ($k$-WL)} is a~standard
heuristic in~the~context of~graph isomorphism testing (see, e.g.,
\cite{Babai16,cai_optimal_1992,GroheN21,Kiefer20}) which tries to distinguish
between two input graphs by propagating information on the~local structure
of~\mbox{$k$-tuples} of~vertices. In~their seminal work, Cai, Fürer and Immerman
\cite{cai_optimal_1992} show that dimension $\Omega(n)$ is required for WL to
distinguish between all non-isomorphic pairs of~$n$-vertex graphs $G,H$
(in~particular, $k$-WL is an incomplete isomorphism test for every $k \geq 1$).
However, the~graph pairs constructed by Cai, Fürer and Immerman only have small
edit distance, and thus are not a~valid input for $\epsilon$-\textsc{GI} (and
the same holds for other known constructions). In~the~following, we show that
much better upper bounds on the~WL dimension can be achieved if we are
guaranteed that the~inputs graphs have large edit distance.

\subsection{The Weisfeiler-Leman Algorithm}

We start by giving a~brief description of~the~WL algorithm. Fix some $k \geq 1$
and let $G$ be a~graph. For colourings
$\lambda_1,\lambda_2\colon (V(G))^k \to C$, we say that $\lambda_1$
\emph{refines} $\lambda_2$, denoted by $\lambda_1 \preceq \lambda_2$,
if $\lambda_2(\bar v) = \lambda_2(\bar w)$ for all $\bar v,\bar w \in (V(G))^k$
for which $\lambda_1(\bar v) = \lambda_1(\bar w)$.
The two colourings are \emph{equivalent}, denoted $\lambda_1 \equiv \lambda_2$,
if $\lambda_1 \preceq \lambda_2$ and $\lambda_2 \preceq \lambda_1$.

The $k$-WL algorithm iteratively computes finer and finer colourings
of $k$-tuples of~vertices.
For $i \geq 0$ we describe the~colouring $\WLit{k}{i}{G}\colon (V(G))^k \to C$
computed in~the~$i$-th round of~the~algorithm.
For $i = 0$, each tuple is coloured with the~isomorphism type of~the~underlying
ordered subgraph.
More formally, if $H$ is a~second graph, $\bar v = (v_1,\dots,v_k) \in (V(G))^k$
and \mbox{$\bar w = (w_1,\dots,w_k) \in (V(H))^k$}, then
$\WLit{k}{i}{G}(\bar v) = \WLit{k}{i}{H}(\bar w)$ if and only if, for every
$i,j \in [k]$, it holds that $v_i = v_j \iff w_i = w_j$ and
$v_iv_j \in E(G) \iff w_iw_j \in E(H)$.

Suppose $k \geq 2$. For $i > 0$, we set
$\WLit{k}{i}{G}(\bar v) \coloneqq
    \big(\WLit{k}{i-1}{G}(\bar v);\mathcal{M}_i(\bar v)\big)$ where
\[\mathcal{M}_i(\bar v) \coloneqq
      \Big\{\!\Big\{ \big(\WLit{k}{i-1}{G}(\bar v[w/1]),\dots,
          \WLit{k}{i-1}{G}(\bar v[w/k])\big) ~\Big|~ w \in V(G)\Big\}\!\Big\}\]
and $\bar v[w/i]$ denotes the~$k$-tuple
$(v_1,\dots,v_{i-1},w,v_{i+1},\dots,v_k)$
obtained from $\bar v$ by replacing the~$i$-th entry with $w$.
For $k = 1$, the~definition is analogous, but we set
\[\mathcal{M}_i(v) =
      \Big\{\!\Big\{ \WLit{k}{i-1}{G}(w) ~\Big|~ w \in N_G(v)\Big\}\!\Big\}.\]
By definition, we obtain that $\WLit{k}{i+1}{G} \preceq \WLit{k}{i}{G}$ for
every $i \geq 0$.
Hence, there is a~minimal $i_\infty$ such that
$\WLit{k}{i_\infty}{G} \equiv \WLit{k}{i_\infty+1}{G}$, and we denote
$\WL{k}{G} \coloneqq \WLit{k}{i_\infty}{G}$ the~\emph{$k$-stable colouring}.
The $k$-dimensional Weisfeiler-Leman algorithm takes as input a~graph $G$ and
outputs (a colouring that is equivalent to) $\WL{k}{G}$.
This can be done in~time $O(k^2n^{k+1} \log n)$ \cite{ImmermanL90}.
 
Let $H$ be a~second graph.
The $k$-dimensional Weisfeiler-Leman algorithm \emph{distinguishes} $G$ and $H$
if there is a~colour $c \in C$ such that
\[\Big|\Big\{ \bar v \in (V(G))^k \Bigmid \WL{k}{G}(\bar v) = c \Big\}\Big|
\neq \Big|\Big\{\bar w \in (V(H))^k \Bigmid \WL{k}{H}(\bar w) = c \Big\}\Big|.\]

\paragraph{Bijective $k$-pebble game.}

We will need a~well-known game characterisation of~$k$-WL
\cite{cai_optimal_1992,Hella96}.

Let $k \geq 1$.
For graphs $G$ and $H$ on the~same number of~vertices, we define the
\emph{bijective $k$-pebble game} $\BP_{k}(G,H)$ as follows:
\begin{itemize}
 \item The~game has two players called \emph{Spoiler} and \emph{Duplicator}.
 \item The~game proceeds in~rounds, each of~which is associated with a~pair
       of~positions $(\bar v,\bar w)$ with~$\bar v \in \big(V(G)\big)^k$
       and~$\bar w \in \big(V(H)\big)^k$.
 \item To determine the~initial position, Duplicator plays a~bijection
       $f\colon \big(V(G)\big)^k \rightarrow \big(V(H)\big)^k$ and Spoiler
       chooses some $\bar v \in \big(V(G)\big)^k$.
       The~initial position of~the~game is then set to $(\bar v,f(\bar v))$.
 \item Each round consists of~the~following steps.
  Suppose the~current position of~the~game is
  $(\bar v,\bar w) = ((v_1,\ldots,v_k),(w_1,\ldots,w_k))$.
  \begin{itemize}
   \item[(S)] Spoiler chooses some $i \in [k]$.
   \item[(D)] Duplicator picks a~bijection $f\colon V(G) \rightarrow V(H)$.
   \item[(S)] Spoiler chooses $v \in V(G)$ and sets $w \coloneqq f(v)$.
      Then the~game moves to position $\big(\bar v[i/v], \bar w[i/w]\big)$
      where $\bar v[i/v] \coloneqq (v_1,\dots,v_{i-1},v,v_{i+1},\dots,v_k)$
      is the~tuple obtained from $\bar v$ by replacing the~$i$-th entry by $v$.
  \end{itemize}

  If mapping each $v_i$ to $w_i$ does not define an isomorphism of~the~induced
  subgraphs of~$G$ and $H$, Spoiler wins the~play.
  More precisely, Spoiler wins if there are~$i,j\in [k]$ such
  that~$v_i = v_j \nLeftrightarrow w_i =w_j$
  or~$v_iv_j \in E(G) \nLeftrightarrow w_iw_j \in E(H)$.
  If vertices are coloured, then $v_i$ and $w_i$ also need to receive the~same
  colour for every $i \in [k]$.
  If there is no position of~the~play such that Spoiler wins, then Duplicator
  wins.
\end{itemize}

We say that Spoiler (and Duplicator, respectively) \emph{wins $\BP_k(G,H)$} if
Spoiler (and Duplicator, respectively) has a~winning strategy for the~game.
Also, for a~position $(\bar v,\bar w)$ with $\bar v \in \big(V(G)\big)^k$ and
$\bar w \in \big(V(H)\big)^k$, we say that Spoiler (and Duplicator,
respectively) \emph{wins $\BP_k(G,H)$ from position $(\bar v,\bar w)$}
if Spoiler (and Duplicator, respectively) has a~winning strategy for the~game
started at position $(\bar v,\bar w)$.

\begin{theorem}[\cite{cai_optimal_1992,Hella96}]
 \label{thm:eq-wl-pebble}
 Suppose $k \geq 2$ and let $G$ and $H$ be two graphs.
 Then $k$-WL distinguishes $G$ and $H$ if and only if Spoiler wins the~game
 $\BP_{k+1}(G,H)$.
\end{theorem}

\subsection{Upper Bounds on the~WL Dimension}
\label{sec:wl-upper}

Recall that, by the~seminal work of~Cai, Fürer and Immerman
\cite{cai_optimal_1992}, WL dimension $\Omega(n)$ is~required for WL to
distinguish between all non-isomorphic pairs of~$n$-vertex graphs $G,H$.
In fact, this even holds true if we restrict to graphs of~maximum degree $3$,
and hence to graphs of~VC~dimension at most $3$.
In the~following, we show that much better upper bounds on the~WL dimension can
be achieved if we are guaranteed that the~edit distance between the~two graphs
$G,H$ is large.

Towards this end, we first require some additional notation.
Let $G$ be a~graph.
For a~tuple $\bar s = (s_1,\dots,s_\ell) \in (V(G))^\ell$ of~pairwise distinct
vertices we write $(G,\bar s)$ to denote the~coloured graph obtained by
individualising vertices in~$\bar s$.
More formally, we colour the~vertices of~$G$ using the~colouring $\lambda$
defined via $\lambda(s_i) \coloneqq i$ for every $i \in [\ell]$, and
$\lambda(v) = 0$ for all $v \in V(G) \setminus \{s_1,\dots,s_\ell\}$.
Usually, the~ordering of~the~elements in~$\bar s$ is not relevant for our
purposes.
Hence, slightly abusing notation, for $S \subseteq V(G)$, we also write $(G,S)$
for the~graph $(G,\bar s)$ where $\bar s$ is any tuple containing each element
of $S$ exactly once.
We also write $\gamma^{G,S} \coloneqq \WL{1}{(G,S)}$ we denote the~colouring
computed by $1$-WL on the~graph $(G,S)$.
If the~graph $G$ is clear from context, we usually omit it and simply write
$\gamma^S$.

The next definition introduces the~key notion for our upper-bound results.
For $v,w \in V(G)$, we define $M_G(v,w) \coloneqq N_G(v) \symdiff N_G(w)$ to be
the \emph{mixed neighbourhood} of~$v,w$.

\begin{definition}
 Let $G$ be an $n$-vertex graph and $\epsilon > 0$.
 A~set $S \subseteq V(G)$ is \textit{$\epsilon$-homogenising} if
 \begin{equation}
  \label{eq:epsilon_homogenising}
  \gamma^{S}(v) = \gamma^{S}(w)\implies |M_G(v,w)| \le \epsilon n.
 \end{equation}
 for all $v,w \in V(G)$.
\end{definition}

The following lemma is the~key tool to obtain our upper bounds.

\begin{lemma}
  \label[lemma]{if_homogeneous_then_small_edit}
  Let $G,H$ be $n$-vertex graphs, $\epsilon > 0$, and $S \subseteq V(G)$
  a~\mbox{$\left(\epsilon/2\right)$-homogenising} set in~$G$.
  If $(|S|+1)$-WL does not distinguish $G$ and $H$, then
  $\ged(G,H) \leq \epsilon n^2$.
\end{lemma}
\begin{proof}
  Let $C_1,\dots,C_\ell$ denote the colour classes of $\gamma^{G,S}$.
  Since $(|S|+1)$-WL does not distinguish $G$ and $H$, there is a set $S' \subseteq V(H)$ such that colour refinement does not distinguish $(G, S)$ and $(H, S')$.
  Let $C'_1,\ldots,C'_{\ell'}\subseteq V(H)$ denote the corresponding~colour classes of~$\gamma^{H,S'}$.
  Note that $\ell' = \ell$ and $|C_i| = |C'_i|$ for all $i \in [\ell]$.
  Furthermore, by the properties of colour refinement,
  there are integers $d_{ij}$ such that every vertex in $C_i$
  has exactly $d_{ij}$ neighbours in $C_j$ for any $i,j\in [\ell]$.
  Let $\pi\colon V(G) \to V(H)$ be an arbitrary colour preserving bijection,
  that is, $\pi(C_i) = C_i'$ for all $i \in [\ell]$.
  We will show that for every $v \in V(G)$,
  \[ |\{w \in V(G) \mid vw\in E(G) \iff \pi(v)\pi(w) \not\in E(H) \}|
      \leq \epsilon n,\]
  which immediately implies $\ged(G,H) \leq \epsilon n^2$.

  Therefore, for the rest of the proof, we fix an arbitrary vertex $v \in V(G)$
  and let $j \in [\ell]$ be the unique index such that $v \in C_j$.
  We define
  \begin{align*}
    I &\coloneqq \{i \in [\ell] \mid
                    0 < d_{ji} < |C_i|\}.\\
    I_+ &\coloneqq \left\{i \in I \Bigmid
                    d_{ji} \leq \frac12 |C_i|\right\}\\
    I_- &\coloneqq I \setminus I_+.
  \end{align*}
  For every $w \in C_j$ we define
  \[R_G(w) \coloneqq 
      \bigcup_{i \in I_+} \left(C_i \cap N_G(w)\right) 
        \quad\cup\quad
      \bigcup_{i \in I_-} \left(C_i \setminus N_G(w)\right).\]
  Intuitively, the set $R_G(w)$ consists of those vertices $x \in C_i$ (for $i \in I$)
  whose adjacency to $w$ disagrees with the majority adjacency between
  $w$ and the colour class $C_i$.
  Observe that for all $w \in C_j$, we have
  \begin{align}
    |R_G(w)| &= |R_G(v)|,\label{eq:equality-of-R}\\
    M_G(v,w) &= N_G(v) \symdiff N_G(w) = R_G(v) \symdiff R_G(w)
        \label{eq:M-in-terms-of-R}
  \end{align}

  \begin{claim}
    \label{claim:1}
    Let $i\in I$ and let $x \in C_i$. Then
    \[|\{ w\in C_j \mid x \in R_G(w)\}| \leq \frac12 |C_j|.\]
  \end{claim}
  \begin{claimproof}
      Fix $x \in C_i$.
      If $i \in I_+$, then $d_{ji} \leq \frac12 |C_i|$ and hence, by definition
      of $R_G$, we have
      \[|\{ w \in C_j \mid x \in R_G(w)\}| = d_{ij}
            = d_{ji} \frac{|C_j|}{|C_i|} \leq \frac12 |C_j|.\]
      Similarly, if $i \in I_-$, then
      \[|\{ w \in C_j \mid x \in R_G(w)\}| = |C_j| - d_{ij}
            = |C_j| - d_{ji} \frac{|C_j|}{|C_i|} < |C_j| - \frac12 |C_j|
            = \frac12 |C_j|.\]
  \end{claimproof}

  \begin{claim}
    \label{claim:existence-w}
    There is a $w \in C_j$ with
    $|R_G(w) \cap R_G(v)| \leq |R_G(w) \setminus R_G(v)|$.
  \end{claim}
  \begin{claimproof}
    Suppose for contradiction, that for all $w \in C_j$, we have
    \[ |R_G(w) \cap R_G(v)| > |R_G(w) \setminus R_G(v)|.\]
    In particular, we have
    \[ \frac12 \sum_{w \in C_j} |R_G(w)| < \sum_{w\in C_j} |R_G(w) \cap R_G(v)|.\]
    Furthermore, by \cref{claim:1},
    \[
      \sum_{w\in C_j} |R_G(w) \cap R_G(v)|
          = \sum_{x\in R_G(v)} |\{ w\in C_j \mid x \in R_G(w) \}|
          \leq \frac12 |C_j| |R_G(v)|.\\
    \]
    Together, we have 
    \[ \sum_{w\in C_j} |R_G(w)|
          < 2\sum_{w\in C_j} |R_G(w) \cap R_G(v)|
          \leq |C_j| |R_G(v)|,\]
    which is a contradiction by \eqref{eq:equality-of-R}.
  \end{claimproof}

  \begin{claim}
    $|R_G(v)| \leq \frac{1}{2}\epsilon n$.
  \end{claim}
  \begin{claimproof}
    Let $w\in C_j$ be as given by \cref{claim:existence-w}. Then
    \begin{align*}
      |R_G(v)|
        &\leq |R_G(v)| + |R_G(w) \setminus R_G(v)| - |R_G(w) \cap R_G(v)|\\
        &=   |R_G(v) \triangle R_G(w)| \\
        &= |M_G(v,w)|  &\eqref{eq:M-in-terms-of-R}\\
        &\leq \frac12 \epsilon n,
    \end{align*}
    as $\gamma^{G,S}(v) = \gamma^{G,S}(w)$ and $S$
    is $(\epsilon/2)$-homogenising.
  \end{claimproof}

  We now want to bound the quantity
  $|\{w \in V(G) \mid vw\in E(G) \iff \pi(v)\pi(w)\notin E(H)\}|$.
  Observe that outside the sets $R_G(v)$ and $R_H(\pi(v))$, 
  the adjacency between $v$ and $w$ is determined by the majority behaviour 
  within the corresponding colour classes, and hence the adjacency must agree.
  Any disagreement can hence only occur if either
  $w \in R_G(v)$ or $\pi(w) \in R_H(\pi(v))$. Furthermore, since colour
  refinement does not distinguish $(G,S)$ and $(H,S')$, we have
  $|R_H(\pi(v))| = |R_G(v)| \leq \tfrac{\epsilon}{2} n$, and hence
  \[
  |\{w \in V(G) \mid vw\in E(G) \iff \pi(v)\pi(w)\notin E(H)\}|
  \leq |R_G(v)| + |R_H(\pi(v))| \leq \epsilon n.
  \]
\end{proof}

In order to construct an $\epsilon$-homogenising set for a~graph $G$, we
consider the~set system \mbox{$\mathcal{M}_G \coloneqq \{M_G(v,w) \mid
v,w\in V(G)\}$} of~mixed neighbourhoods.
Now, the~key observation is that every \mbox{$\epsilon$-net} $S$ for
$\mathcal M_G$ is \mbox{$\epsilon$-homogenising.}
Indeed, if $\gamma^S(v) = \gamma^S(w)$, then, in~particular,
\mbox{$S \cap M_G(v,w) = \emptyset$}, which by the~definition
of~an~\mbox{$\epsilon$-net} implies $|M_G(v,w)| \leq \epsilon n$.
Hence, with \cref{lem:epsilon-net} in~mind, it only remains to bound the
VC~dimension of~$\mathcal M_G$. We can get the following bound as a direct
consequence of \cref{lem:general-10d-lemma}.

\begin{lemma}
 \label[lemma]{lem:vc-dim-mixed}
 Let $G$ be a~graph of~VC~dimension $d$.
 Then $\mathcal{M}_G$ has VC~dimension at most $10d$.
\end{lemma}

\begin{theorem}
  \label{thm:wl-solves-bounded-vc-robust-iso}
  The~$O\left(\frac{d}{\epsilon} \log \frac{1}{\epsilon}\right)$-WL
  distinguishes all pairs of~non-isomorphic $n$-vertex graphs $G,H$ such that
  $G$ has VC~dimension at most $d$ and $\ged(G,H) \geq \epsilon n^2$.
\end{theorem}
\begin{proof}
 Let $\epsilon' \coloneqq \frac{1}{3}\epsilon$.
 Let $G,H$ be non-isomorphic $n$-vertex graphs such that $G$ has VC~dimension
 at most $d$ and $\ged(G,H) \geq \epsilon n^2$.
 By \cref{lem:epsilon-net,lem:vc-dim-mixed}, there is an $\epsilon'$-net $S$
 for the~set system $\mathcal{M}_G$ of~size
 $O(\frac{d}{\epsilon'}\log \frac{1}{\epsilon'})
      = O(\frac{d}{\epsilon}\log \frac{1}{\epsilon})$.
 Since $S$ is $\epsilon'$-homogenising, we conclude by
 \cref{if_homogeneous_then_small_edit} that $(|S|+1)$-WL distinguishes $G,H$.
\end{proof}

Replacing \cref{lem:epsilon-net} by \cref{lem:epsilon-net-general} in~the~last
proof, we obtain a~slightly improved bound for general $n$-vertex graphs.

\begin{theorem}
  \label{thm:log-wl-solves-robust-iso}
  The~$O\left(\frac{1}{\epsilon}\log n\right)$-WL algorithm distinguishes
  all pairs of~non-isomorphic $n$-vertex graphs $G,H$ such that
  $\ged(G,H) \geq \epsilon n^2$.
\end{theorem}

Finally, for coloured graphs of~colour class size at most $s$, we can also find
$\epsilon$-homogenising~sets of~small size via a~different argument.
In this case, we proceed inductively: if a~set $S$ is not
$\epsilon$-homogenising because $|M_G(v,w)| > \epsilon n$ for some $v,w$,
individualising $S$ and $v$ splits a~linear number of~colour classes. Since
such a~split may happen only constant number of~times, we get the~existence
of a~small $\epsilon$-homogenising set.

\begin{lemma}
 \label[lemma]{lem:homogeneity-for-bounded-colour_class}
 Let $(G,\lambda_G)$ be a~coloured graph such that all colour classes have size
 at most $s$. Then there is an $\epsilon$-homogenising set $S \subseteq V(G)$
 of~size $|S| \leq \frac{s-1}{\epsilon}$.
\end{lemma}
\begin{proof}
  Let $S_0 = \emptyset$.
  Consider some $i \geq 0$ for which $S_i$ is defined, and suppose that $S_i$ is not \mbox{$\epsilon$-homogenising}.
  Then there are $v,w$ such that $\gamma^{S_i}(v) = \gamma^{S_i}(w)$ and
  $|M_G(v,w)| > \epsilon n$.
  We set $S_{i+1} \coloneqq S_i \cup \{w\}$.

  Let $\ell$ denote the index at which the iterative process is stopped, i.e., the set $S_\ell$ is $\epsilon$-homogenising.
  We argue that $\ell \leq \frac{s-1}{\epsilon}$.
  Since only one element is added in each iteration, this implies the desired bound on the size of $S$.

  \begin{claim}
    \label[claim]{claim:individualising_large_degree_splits_colour_classes}
    Let $0 \leq i < \ell$, and let $v,w$ be a~pair such that $\gamma^{S_i}(v) = \gamma^{S_i}(w)$.
    Let $C$ be a~colour class of~$\gamma^{S_i}$ such that
    $M_G(v,w) \cap C \neq \emptyset$.
    Then there are $x,y \in C$ such that
    $\gamma^{S_{i+1}}(x) \neq \gamma^{S_{i+1}}(y)$.
  \end{claim}
  \begin{claimproof}
    Let $v,w$ be such a~pair and $C$ such a~colour class.
    Let $x\in M_G(v,w)\cap C$.
    Without loss of~generality, suppose that
    $x \in N_G(v)\setminus N_G(w)$. By the~properties of~colour
    refinement, there is a~$y\in N_G(w)\cap C$ and hence $y \neq x$.
    But this means that by individualising $w$, the~vertices $x$ and $y$
    get different colours, that is,
    $\gamma^{S_{i+1}}(x) \neq \gamma^{S_{i+1}}(y)$.
  \end{claimproof}

  \begin{claim}
    For every $i \in [\ell]$, there are at least $\frac{1+i\epsilon}{s}n$ colour classes in~$\gamma^{S_i}$.
  \end{claim}
  \begin{claimproof}
    We prove the statement by induction on $i$. For $i = 0$, since each colour class has size
    at most $s$, their number is at least $n/s$.

    Suppose that the~statement holds for $i < \ell$.
    Consider the~set $S_i$ and let $v,w$ be a~pair
    such that $\gamma^{S_i}(v) = \gamma^{S_i}(w)$, $|M_G(v,w)| > \epsilon n$
    and $w \in S_{i+1}\setminus S_i$.

    Clearly, there are at least $\frac1s |M_G(v,w)|$ colour classes $C$
    (of the~colouring $\gamma^{S_i}$) such that $C \cap M_G(v,w) \neq \emptyset$.
    Hence, by \cref{claim:individualising_large_degree_splits_colour_classes},
    individualising $w$ splits $\frac1s |M_G(v,w)|$ colour classes.
    This means that $\gamma^{S_{i+1}}$ has at least
    $\frac{1+i\epsilon}{s}n + \frac1s |M_G(v,w)| \geq \frac{1+(i+1)\epsilon}{s}n$ colour classes.
  \end{claimproof}

  Since there can be at most $n$ colour classes in $\gamma^{S_\ell}$, we conclude that $\ell \leq \frac{s-1}{\epsilon}$ as desired.
\end{proof}

Using the~same arguments as before, we obtain the~following.
\begin{theorem}
    \label{thm:wl-solves-bounded-colour-class-size-robust-iso}
  The~$O(\frac{s}{\epsilon})$-WL algorithm distinguishes all
  pairs of~non-isomorphic $n$-vertex graphs $G,H$ of~colour class size at most
  $s$ such that $\ged(G,H) \geq \epsilon n^2$.
\end{theorem}

\subsection{Lower bounds on the~WL dimension}
\label{sec:wl-lower}

We complement our upper bound results for the~WL algorithm with some lower
bounds, which in~particular show that the~dependence on $\epsilon$ on the~above
results cannot be improved by much.
Our lower bounds build on existing constructions from
\cite{cai_optimal_1992,OdonnellWWZ14}.
While the~pairs of~graphs constructed in~\cite{cai_optimal_1992,OdonnellWWZ14}
only have small edit distance, we can use blowups to increase their edit
distance.
Let $G = (V,E,\lambda)$ be a~coloured graph and $\ell \geq 1$.
We define the~\emph{$\ell$-blowup} of~$G$ to be the~graph $G^{(\ell)}$ obtained
by replacing each edge of~$G$ by a~complete bipartite graph $K_{\ell,\ell}$.
More formally, $G^{(\ell)}$ has vertices
$V(G^{(\ell)}) = \{(v, i) \mid v \in V, i \in [\ell]\}$ and edges
$E(G^{(\ell)}) = \{ (v,i)(w,j) \mid vw \in E, i,j \in [\ell]\}$.
We equip $G^{(\ell)}$ with the~colouring $\lambda^{(\ell)}$ defined via
$\lambda^{(\ell)}(v,i) = (\lambda(v),i)$.
The following lemma states that large blowups have quadratic edit distance.

\begin{lemma}[\cite{pikhurko_analytic_2010}]
  \label{lem:edit-distance-blowup}
  Let $G, H$ be coloured graphs and $\ell \geq 1$. Then
  $\ged(G^{(\ell)},H^{(\ell)}) \geq \frac13 \cdot \ell^2 \cdot \ged(G,H)$.
\end{lemma}

Also, taking the~blowup of~a graph does not affect the~power of~the~WL algorithm.

\begin{lemma}
  \label{lem:wl-blowup}
  Let $G, H$ be coloured graphs and $\ell \geq 1$, $k \geq 2$.
  Then $k$-WL distinguishes $G$ and $H$ if and only if $k$-WL distinguishes
  $G^{(\ell)}$ and $H^{(\ell)}$.
\end{lemma}
\begin{proof}
  First suppose that $k$-WL does not distinguish $G$ and $H$, i.e., Duplicator
  has a~winning strategy for $\BP_{k+1}(G,H)$.
  We give a~winning strategy for Duplicator in~$\BP_{k+1}(G^{(\ell)},H^{(\ell)})$.

  Let $f\colon (V(G))^{k+1} \to (V(H))^{k+1}$ be the~bijection chosen by
  Duplicator in~the~initial round of~$\BP_{k+1}(G,H)$.
  We define the~bijection
  $f^{(\ell)}\colon (V(G^{(\ell)}))^{k+1} \to (V(H^{(\ell)}))^{k+1}$ via
  \[f^{(\ell)}((v_1,j_1),\dots,(v_{k+1},j_{k+1}))
            = ((w_1,j_1),\dots,(w_{k+1},j_{k+1}))\]
  where $(w_1,\dots,w_{k+1}) = f(v_1,\dots,v_{k+1})$.
  Duplicator initially plays the~bijection $f^{(\ell)}$ in~the~game
  $\BP_{k+1}(G^{(\ell)},H^{(\ell)})$.
  Throughout the~game, Duplicator maintains the~invariant that only positions
  of~the~form
  \[\Big(\big((v_1,j_1),\dots,(v_{k+1},j_{k+1})\big),
         \big((w_1,j_1),\dots,(w_{k+1},j_{k+1})\big)\Big)\]
  are reached, where $((v_1,\dots,v_{k+1}),(w_1,\dots,w_{k+1}))$ is a~winning
  position in~the~game $\BP_{k+1}(G,H)$.
  Clearly, this is the~case after the~initial round.

  We need to argue that the~invariant can be maintained.
  So suppose the~game is in~such a~position, and Spoiler chooses $i \in [k+1]$
  in~the~game $\BP_{k+1}(G^{(\ell)},H^{(\ell)})$.
  Let $f\colon V(G) \to V(H)$ be Duplicator's response in~the~game
  $\BP_{k+1}(G,H)$ if Spoiler chooses the~same index $i$ in~the~position
  $((v_1,\dots,v_{k+1}),(w_1,\dots,w_{k+1}))$.
  Then, in~the~game $\BP_{k+1}(G^{(\ell)},H^{(\ell)})$, Duplicator responds
  with the~bijection $f^{(\ell)}\colon V(G^{(\ell)}) \to V(H^{(\ell)})$ defined
  via
  \[f^{(\ell)}(v,j) = (f(v),j)\]
  for all $v \in V(G)$ and $j \in [\ell]$.
  It is easy to see that Duplicator maintains the~invariant no matter which
  vertex is chosen by Spoiler to complete the~round.

  Finally, note that Duplicator can never loose if the~invariant is maintained,
  which completes the~first direction.

  Now, suppose that $k$-WL distinguishes $G$ and $H$, i.e., Spoiler has a~winning strategy for $\BP_{k + 1}(G,H)$.
  Let $X_G \coloneqq \{(v,1) \mid v \in V(G)\}$ and $X_H \coloneqq \{(w,1) \mid w \in V(H)\}$.
  Observe that $X_G$ and $X_H$ are unions of corresponding color classes in $G$ and $H$, respectively.
  Moreover, $G^{(\ell)}[X_G] \cong G$ and $H^{(\ell)}[X_H] \cong H$.
  Hence, a winning strategy for Spoiler in $\BP_{k + 1}(G,H)$ immediately translates to a winning strategy in $\BP_{k + 1}(G^{(\ell)},H^{(\ell)})$, by identifying a vertex $v$ in $G$ or $H$ with the vertex $(v,1)$ in $G^{(\ell)}$ or $H^{(\ell)}$, respectively.
  Observe that, due the coloring of vertices, Spoiler immediately wins whenever Duplicator chooses a bijection that does not bijectively map vertices from $W_G$ to $W_H$.
\end{proof}

Now, we obtain the~following using the~standard CFI construction from
\cite{cai_optimal_1992}.

\begin{theorem}
 \label{thm:wl-lower-bound-via-cfi}
 For every $k \geq 1$, $n \geq k$ there are non-isomorphic (coloured) graphs
 $G,H$ such that
 \begin{itemize}
  \item $|V(G)| = |V(H)| = \Theta(n)$,
  \item $k$-WL does not distinguish $G,H$,
  \item $\ged(G,H) \geq \frac{n^2}{k^2}$,
  \item $G$ and $H$ have VC~dimension at most $3$, and
  \item each colour class in~$G$ and $H$ has size at most $4$.
 \end{itemize}
\end{theorem}
\begin{proof}
  Let $k \geq 1$.
  By \cite{cai_optimal_1992}, there are non-isomorphic 3-regular coloured graphs
  $G',H'$ of~size $\Theta(k)$, with colour class size at most $4$, that are
  not distinguished by $k$-WL.
  In~particular, $\ged(G', H') \geq 1$.

  Let $n \geq k$ and set $\ell \coloneqq \ceil{3n/k}$.
  Consider the~graphs $G \coloneqq (G')^{(\ell)}$ and
  $H \coloneqq (H')^{(\ell)}$.
  Then $|V(G)| = |V(H)| = \Theta(n)$, and $k$-WL does not distinguish
  $G$ and $H$ by \cref{lem:wl-blowup}.
  Also, by \cref{lem:edit-distance-blowup}, we have that
  \[\ged(G,H) \geq \frac13 \cdot \ell^2 \cdot \ged(G',H')
              \geq \frac{n^2}{k^2}.\]

  Let $S \subseteq V(G)$ be a~set shattered by $\mathcal N_G$. Clearly, we
  cannot have $(v,i), (v,j) \in S$ for some $v \in V(G)$
  and $i\neq j \in [n]$, since the~two vertices have exactly the~same
  neighbours.
  Hence the~set
  $S' = \{v \in V(G') \mid (v,i) \in V(G) \text{ for some }i\in [n]\}$ has the
  same cardinality as $S$.
  Furthermore, since $S$ is shattered by $\mathcal N_G$, the~set $S'$ is
  shattered in~$\mathcal N_{G'}$.
  However, since every vertex in~$G'$ has at most degree $3$, we conclude that
  $|S| = |S'| \leq 3$.
  The~same reasoning holds true for $H$.

  Lastly, the~colour class size of~$G'$ and $H'$ is at most $4$, and the
  definition of~$n$-th blowup does not increase the~colour class sizes.
\end{proof}

In particular, the~theorem implies that $o(1/\sqrt{\epsilon})$-WL cannot
distinguish between graphs of~edit distance $\epsilon n^2$, even for graphs
of constant VC~dimension and constant colour class size.
We can improve on the~dependence on $\epsilon$ by instead using a~construction
from \cite{OdonnellWWZ14}.

\begin{theorem}
 \label[theorem]{thm:linear_lower_bound}
 For every $k \geq 1$, $n \geq k$ there are non-isomorphic graphs $G,H$ such
 that
 \begin{itemize}
  \item $|V(G)| = |V(H)| = \Theta(n)$,
  \item $k$-WL does not distinguish $G,H$, and
  \item $\ged(G,H) \geq \frac{n^2}{k}$.
 \end{itemize}
\end{theorem}
\begin{proof}
 Let $k \geq 1$.
 By \cite{OdonnellWWZ14,AtseriasM13,roberson_lasserre_2024}, there are graphs
 $G',H'$ of~size $\Theta(k)$ that are not distinguished by
 $k$-WL, but $\ged(G', H') \geq k$.

 Let $n \geq k$ and set $\ell \coloneqq \ceil{3n/k}$.
 Consider the~graphs $G \coloneqq (G')^{(\ell)}$ and
 $H \coloneqq (H')^{(\ell)}$.
 Then $|V(G)| = |V(H)| = \Theta(n)$, and $k$-WL does not distinguish
 $G$ and $H$ by \cref{lem:wl-blowup}.
 Also, by \cref{lem:edit-distance-blowup}, we have that
 \[\ged(G,H) \geq \frac13 \cdot \ell^2 \cdot \ged(G',H')
             \geq \frac{n^2}{k}.\qedhere\]
\end{proof}

In particular, we get that $o(1/\epsilon)$-WL cannot distinguish between graphs
of~edit distance~$\epsilon n^2$.

\section{Concluding Remarks}
We identified VC dimension as a~key parameter governing the~runtime
of~additive approximation algorithms for the~graph edit distance problem
and the~quadratic assigment problem: for a~fixed approximation error
$\epsilon$ the~runtime is $n^{O_\epsilon(d)}$, where $d$ is
an appropriately defined VC dimension of the~input instance.
Though we did not discuss them here except briefly in
Remark~\ref{rem:higher-order}, similar results can be obtained
for related assignement problems (see~\cite{AroraFK02}).

To solve the~robust graph isomorphism problem, which may be viewed as
a~special case of graph edit distance, we can even use the
$O_\epsilon(d)$-dimensional WL algorithm, which implies that for all
graphs the~$O_\epsilon(\log n)$-dimensional WL algorithm suffices to
solve $\epsilon$-\GI, whereas exact graph isomorphism
requires WL dimension $\Omega(n)$.

The question whether for every $\epsilon$ there is a~polynomial time
algorithm solving $\epsilon$-$\GI$ remains open. We believe this
a~natural and interesting question that has received surprisingly
little attention. Related to this is the~question whether for all
$\epsilon$ there is a~$k$ (not depending on the~order $n$ of the~input
graphs or their VC dimension $d$) such that $k$-WL solves
$\epsilon$-\GI. Note that our lower bound says that for every $k$ there
is an $\epsilon$ such that $k$-WL does not solve $\epsilon$-\GI.

\printbibliography

\appendix

\section{Assignment problems}
\label{app:assignment-problem}

Let us consider the following generalisation of~the assignment problem,
called \textit{Assignment Problem with Extra Constraints}, or APEC:
\begin{align}
  \label{apec}
  x \in \mathcal A~&  \\
  a_k x \geq b_k & \hspace{15pt} k \in [K], \nonumber
\end{align}
where $\mathcal A$ is the standard set of~matching constraints:
$\sum_i x(i,j) = 1$, $\sum_j x(i,j) = 1$, $x(i,j) \geq 0$.
Such APEC is called $B$-bounded if all coefficients $a_k$ are between $-B$
and $B$.

The following theorem is a~subcase of~\cite[Theorem 3]{AroraFK02} for
$B$-bounded APEC. They show a~randomised algorithm
in \cite[Section 2]{AroraFK02}
and state it is possible to derandomise it in~a~standard way
using the method of~conditional probabilities.
\begin{theorem}[\cite{AroraFK02}]
  \label{thm:arora-rounding-theorem}
  Let $B$ be constant. There is an algorithm that, given a~$B$-bounded APEC
  as in~\eqref{apec} and a~fractional solution $x^*$, produces a~matching $x$
  that contains at~least $n - o(n)$ edges that satisfies
  \[ a_k x \geq (1 - o(1))b_k - \tilde O(B\sqrt n),\]
  The running time of~the algorithm is $\tilde O(N)$, where $N$ is the number
  of~non-zero entries in~$x^*$.
\end{theorem}

Using this theorem, we can prove \cref{arora-lemma}.

\begin{proof}[Proof of~\Cref{arora-lemma}]
    The proof follows \cite{AroraFK02}. Suppose a~$B$-bounded QAP,
    and let $x^*$ be the optimal (integer) solution,
    $\phi^* = \phi_{x^*}$ be the corresponding optimal bijection.
    We will do the~following~procedure for all partial injective maps
    $\alpha:[n] \to [n]$ of~size at~most $m$.
    For each of~them, the procedure will be in~polynomial time, and
    the overall running time will be in~$n^{O(m)}$.

    Let $\alpha: [n] \to [n]$ be of~size $m$ and let $\epsilon > 0$.
    We construct the following linear program. For any $v,v' \in [n]$,
    we have a~variable $x(v,v')$.
    \begin{align}
      \label{proof-linear-program}
        \text{minimise} \hspace{40pt}  \cost'(x) \coloneqq
                    \sum_{v,v'} \hspace{5pt} &b_{\alpha}(v,v') x(v,v') & \\
        \text{s.t.} \hspace{40pt} x &\in \mathcal A_G & \nonumber \\
        \mathbf{a}(v,v') \cdot x &
            \geq b_{\alpha}(v,v') - \frac13\epsilon n & \forall v,v' \nonumber\\
        \mathbf{a}(v,v') \cdot x &
            \leq b_{\alpha}(v,v') + \frac13\epsilon n & \forall v,v' \nonumber
    \end{align}
    where $\mathcal A_G$ is the set of~assignment constraints
    \mbox{$\sum_{v}x(v,v') = 1$}, \mbox{$\sum_{v'}x(v,v') = 1$},
    \mbox{$x(v,v') \geq 0$},
    and~$\mathbf{a}(v,v')$ is a~vector containing the value $c(v,v',w,w')$
    at~position $w,w'$. We think of~these vectors as
    ``linearised matrices'' and then $\cdot$ is a~scalar product.

    Solve (fractionally) the linear program \eqref{proof-linear-program} and
    denote the value of~the optimal solution by~$\zeta$.
    We create an APEC as follows:
    \begin{align*}
        x &\in \mathcal A_G\\
        \mathbf b \cdot x &\leq \frac{1}{n}\zeta, \nonumber \\
        \mathbf{a}(v,v') \cdot x &
            \geq b_{\alpha}(v,v') - \frac13\epsilon n & \forall v,v' \nonumber\\
        \mathbf{a}(v,v') \cdot x &
            \leq b_{\alpha}(v,v') + \frac13\epsilon n & \forall v,v' \nonumber
    \end{align*}
    where $\mathbf b$ is a~vector containing $\frac{1}{n}b_{\alpha}(i,j)$ at
    position $i,j$.

    By~\Cref{thm:arora-rounding-theorem}, there is a~polynomial time algorithm,
    which with high probability, produces a~matching $\hat{z}$ that contains at
    least $n-o(n)$ edges and satisfies
    \begin{align}
        \label{proof:properties-of-z}
        \hat{z} &\in \mathcal A_G\\
        \mathbf b \cdot \hat{z} &
            \leq (1+o(1))\frac{1}{n}\zeta + \tilde{O}(B\sqrt n),\nonumber\\
        \mathbf{a}(v,v') \cdot \hat{z} &
            \geq (1-o(1)) \left(b_{\alpha}(v,v') - \frac13\epsilon n\right)
                                -\tilde{O}(B\sqrt{n}) & \forall v,v',\nonumber\\
        \mathbf{a}(v,v') \cdot \hat{z} &
            \leq (1+o(1))\left(b_{\alpha}(v,v') + \frac13\epsilon n\right)
                                +\tilde{O}(B\sqrt{n}) & \forall v,v'. \nonumber
    \end{align}
    Note that the maximal coefficient in~$\mathbf b$ is at~most $1$ and
    the maximal coefficient in~$\mathbf a(v,v')$ is at~most $B$.
    Extend $\hat{z}$ arbitrarily to a~perfect matching $z$, adding $o(n)$ edges,
    and output $z$.

    Now, suppose that we have the $\alpha$ with the condition given by
    the lemma, that is, for all $v,v'\in[n]$ we have
    \eqref{eq:approximation-of-b}. Note that
    $\mathbf a(v,v') \cdot x^* = b_{\bar \phi}(v,v')$ and hence,
    by~\eqref{eq:approximation-of-b}, $x^*$ is a~feasible solution to the
    linear program \eqref{proof-linear-program}. Therefore, we get
    \begin{align}
      \label{zeta-bound}
      \zeta \leq \cost'(x^*) = \sum_{v,v'} b_{\alpha}( v,v')x^*(v,v')
            \leq \sum_{v,v'} b_{\bar \phi}(v,v') x^*(v,v') + \frac\epsilon3 n^2
            \leq \cost(x^*) + \frac\epsilon3 n^2.
    \end{align}

    Furthermore, we have that the matching $z$ which
    by~\eqref{proof:properties-of-z} satisfies
    \[\mathbf b \cdot z
            \leq (1 + o(1)) \frac{1}{n} \zeta + \tilde{O}(B \sqrt n) + o(n)\]
    and multiplying by~$n$ we get
    \begin{equation}
        \label{approx-cost-upper-bound}
        \sum_{v,v'} b_{\alpha}(v,v') z(v,v') \leq (1+o(1))\zeta + o(n^2).
    \end{equation}
    Furthermore, \eqref{proof:properties-of-z} implies
    \begin{align}
        \label{upper-bound-bij}
        \sum_{k,l} c(v,v',k,l) z(k,l) &
            \leq (1+o(1))\left(b_{\alpha}(v,v') + \frac13\epsilon n\right)
                        + \tilde{O}(B \sqrt n) + o(Bn) \nonumber \\
        &\leq b_{\alpha}(v,v')+\frac13\epsilon n + o(Bn).
    \end{align}
    We consider $B$ to be a~constant. We can bound the cost of~$z$ as follows:
    \begin{align*}
        \cost(z) &= \sum_{v,v'}\sum_{k,l}  c(v,v',k,l) z(k,l)z(v,v')
                                            & \text{definition of~}\cost\\
        &\leq \sum_{v,v'} (b_{\alpha}(v,v') + \frac13\epsilon n + o(n))z(v,v')
                                            & \eqref{upper-bound-bij}\\
        &\leq \frac13\epsilon n^2 + o(n^2) + \sum_{v,v'} b_{\alpha}(v,v')z(v,v')
                                            &  \\
        &\leq \frac13\epsilon n^2 + (1+o(1))\zeta + o(n^2)
                                            & \eqref{approx-cost-upper-bound}\\
        &\leq \zeta + \frac13\epsilon n^2 + o(n^2) & \\
        &\leq \cost(x^*) + \frac23\epsilon n^2 + o(n^2)
                                            & \text{\eqref{zeta-bound}}\\
        &\leq \cost(x^*) + \epsilon n^2     & \text{for $n$ large enough,}\\
    \end{align*}
    which concludes the proof of~the lemma.
\end{proof}

\end{document}